\documentclass[11pt]{llncs}

\usepackage{paralist}
\usepackage{amssymb,amsmath}
\usepackage{graphicx}
\usepackage{a4wide}

\newcommand{\BO}[1]{{O}\left(#1\right)}
\newcommand{\BT}[1]{{\Theta}\left(#1\right)}
\newcommand{\BOM}[1]{{\Omega}\left(#1\right)}
\newcommand{\T}{\mathcal{T}}
\newcommand{\Proc}{P}
\newcommand{\B}{\mathcal{B}}
\newcommand{\D}[1]{\Delta_{\mathrm{#1}}}

\newcommand{\sel}[1]{\emph{#1}}

\title{Space-Efficient Parallel Algorithms for Combinatorial Search Problems\thanks{An extended abstract~\cite{PietracaprinaPSV13}  of this work was presented at the 38th International Symposium on Mathematical Foundations of Computer Science, 2013.}}

\author{A. Pietracaprina\inst{1} \and G. Pucci\inst{1} \and F. Silvestri \inst{1} \and 
F. Vandin \inst{2,3}}
\institute{Dipartimento di Ingegneria dell'Informazione, University of Padova\\ 
 \email{\{capri,geppo,silvest1\}@dei.unipd.it}
\and 
Department of Mathematics and Computer Science, University of Southern Denmark\\
\texttt{vandinfa@imada.sdu.dk}
\and
Computer Science Department, Brown University}

\begin{document}

\maketitle
\begin{abstract}
We present space-efficient parallel strategies for two fundamental combinatorial search problems, namely, \emph{backtrack search} and \emph{branch-and-bound}, both involving the visit of an $n$-node tree of height $h$ under the assumption that a node can be accessed only through its father or its children.  For both problems we propose efficient algorithms that run on a $p$-processor distributed-memory machine. For backtrack search, we give a deterministic algorithm running in $\BO{n/p+h\log p}$ time, and a Las Vegas algorithm requiring optimal $\BO{n/p+h}$ time, with high probability. Building on the backtrack search algorithm, we also derive a Las Vegas algorithm for branch-and-bound which runs in $\BO{(n/p+h\log p \log n)h\log^2 n}$ time, with high probability. A remarkable feature of our algorithms is the use of only constant space per processor, which constitutes a significant improvement upon previous algorithms whose space requirements per processor depend on the (possibly huge) tree to be explored.
\end{abstract}

\section*{Introduction}

The exact solution of a combinatorial (optimization) problem is often
computed through the systematic exploration of a tree-structured
solution space, where internal nodes correspond to partial solutions
(growing progressively more refined as the depth increases) and leaves
correspond to feasible solutions. A suitable algorithmic
template used to study this type of problems (originally
proposed in \cite{KSW86}) is the exploration of a tree $\T$ under the
constraints that:
\begin{inparaenum}[(i)] 
\item only the tree root is initially known; 
\item the structure, size and height of the tree are unknown;
and \item a tree node can
be accessed if it is the root of the tree or 
if either its father or one of its children is available.
\end{inparaenum} 

In the paper, we focus on two important instantiations of the above
template. The \emph{backtrack search problem} \cite{KZ93} requires to
explore the entire tree $\T$ starting from its root $r$, so to
enumerate all solutions corresponding to the leaves. In the
\emph{branch-and-bound problem}, each tree node is associated to
a cost, and costs satisfy the min-heap order property, so that the cost
of an internal node is a lower bound to the cost of the solutions
corresponding to the leaves of its subtree.  The objective here is to
determine the leaf associated with the solution of minimum cost. We
define $n$ and $h$ to be, respectively, the number of nodes and the height
of the tree to be explored. It is important to remark that in the
branch-and-bound problem, the nodes that must necessarily be
explored are only those whose cost is less than or equal to the cost of
the solution to be determined. These nodes form a subtree $\T^*$ of $\T$
and in this case $n$ and $h$ refer to $\T^*$.  Assuming that a node is
explored in constant time, it is easy to see that the solution to the
above problems requires $\BOM{n}$ time, on a sequential machine, and
$\BOM{n/p+h}$ time on a $p$-processor parallel machine.
 
Due to the elevated computational requirements of search problems,
many parallel algorithms have been proposed in literature that
speed-up the execution by evenly distributing the computation among
the available processing units. All these studies have focused mainly
on reducing the running time while the resulting memory requirements
(expressed as a function of the number of nodes to be stored locally
at each processor) may depend on the tree parameters. However, the
search space of combinatorial problems can be huge, hence it is
fundamental to design algorithms which exploit parallelism to speed up
execution and yet need a small amount of memory per processor,
possibly independent of the tree parameters.  Reducing space
requirements allows for a better exploitation of the memory hierarchy
and enables the use of cheap distributed-memory parallel platforms
where each processing units is endowed with limited memory.

\subsubsection*{Previous work}
Parallel algorithms for backtrack search have been studied in a number
of different parallel models. Randomized algorithms have been
developed for the complete network
\cite{KZ93,LAB93} and the butterfly network \cite{R94}, which require
optimal $\BT{n/p+h}$ node explorations (ignoring the overhead due to
manipulations of local data structures). The work of Herley et
al.~\cite{HPP02} gives a deterministic algorithm running in
$\BO{(n/p+h) (\log \log \log p)^2}$ time on a $p$-processor COMMON CRCW
PRAM. While the algorithm in \cite{KZ93} performs depth-first
explorations of subtrees locally at each processor requiring $\BOM{h}$
space per processor, the other algorithms mostly concentrate on
balancing the load of node explorations among the available processors
but may require $\BOM{n/p}$ space per processor. 

In \cite{KZ93} an $\BT{n/p+h}$-time randomized algorithm for
branch-and-bound is also provided for the complete network.  In
\cite{HPP99ppl,HPP99} Herley et al. show that a parallelization of the
heap-selection algorithm of \cite{F90} gives, respectively, a
deterministic algorithm running in time $\BO{n/p+h\log^2 (np)}$ on an
EREW-PRAM, and one running in time $\BO{(n/p+h\log^4p)\log \log p}$ on
the Optically Connected Parallel Computer (OCPC), a weak variant of
the complete network \cite{AM88}. All of these works adopt a
best-first like strategy, hence they may need $\BOM{n/p}$ space per
processor.  In \cite{KP94} deterministic algorithms for both backtrack
search and branch-and-bound are given which run in $\BO{\sqrt{nh}\log
n}$ time on an $n$-node mesh with constant space per processor.
However, any straightforward implementation of these algorithms on a
$p$-processor machine, with $p < n$, would still require $\BOM{n/p}$
space per processor. Karp et al.~\cite{KSW86} describe sequential
algorithms for the branch-and-bound problem featuring a range of
space-time tradeoffs. The minimum space they attain is $\BO{\sqrt{\log
n}}$ in time $\BO{n 2^{O(\sqrt{\log n})}}$\footnote{The
authors claim a constant-space randomized algorithm running in
$\BO{n^{1+\epsilon}}$ time which, however, disregards the nonconstant
space required by the recursion stack.}.  Some papers (see~\cite{MD99}
and references therein) describe sequential and parallel algorithms
for branch-and-bound with limited space, which interleave depth-first
and breadth-first strategies, but provide no analytical guarantee on
the running time.

\subsubsection*{Our Contribution}\label{sec:contrib}
In this paper, we present space-efficient parallel algorithms for the
backtrack search and branch-and-bound problems. The algorithms are
designed for a $p$-processor distributed-memory message-passing system
similar to the one employed in \cite{KZ93}, where in one time step
each processor can perform $\BO{1}$ local operations and send/receive
a message of $\BO{1}$ words to/from another arbitrary processor. In
case two or more messages are sent to the same processor in one step, we
make the restrictive assumption that none of these messages is
delivered (as in the OCPC model \cite{AM88,GJLR97}). Consistently with most
previous works, we assume that a memory word is sufficient to store a
tree node, and, as in~\cite{KSW86}, we also assume that, given a tree
node, a processor can generate any one of its children or its father
in $\BO{1}$ steps and $\BO{1}$ space.  

For the backtrack search problem we develop a deterministic algorithm
which runs in $\BO{n/p+h\log p}$ time, and a Las Vegas randomized
algorithm which runs in optimal $\BT{n/p+h}$ time with high
probability, if $p = \BO{n/\log n}$. Both algorithms require only
constant space per processor and are based on a nontrivial lazy
implementation of the work-distribution strategy featured in the
backtrack search algorithm by \cite{KZ93}, whose exact implementation
requires $\BOM{h}$ space per processor.  By using the deterministic
backtrack search algorithm as a subroutine, we develop a Las Vegas
randomized algorithm for the branch-and-bound problem which runs in
$\BO{(n/p+h \log p\log n)h\log^2 n}$ time with high probability, using
again constant space per processor.

To the best of our knowledge, our backtrack search algorithms are the
first to achieve (quasi) optimal time using constant
space per processor, which constitutes a significant
improvement upon the aforementioned previous works. As for the branch-and-bound
algorithm, while its running time may deviate substantially from the
trivial lower bound, for search spaces not too deep and sufficiently
high parallelism, it achieves sublinear time using
constant space per processor. For instance, if $h = \BO{n^\epsilon}$
and $p = \BT{n^{1-\epsilon}}$, with $0 < \epsilon < 1/2$, the
algorithm runs in $\BO{n^{2\epsilon} \mbox{polylog}(n)}$ time, with
high probability, using $\BT{n^{1-\epsilon}}$ aggregate space. Again,
to the best of our knowledge, ours is the first algorithm achieving
sublinear running time using sublinear (aggregate) space, thus
providing important evidence that branch-and-bound can be parallelized
in a space-efficient way.

For simplicity, our results are presented assuming that the
tree $\T$ to be explored is binary and that each internal node has
both left and right children. The same results extend to the
case of $d$-ary trees, with $d = \BT{1}$, and to trees that allow
an internal node to have only one child. 


The rest of the paper is organized as follows. In Section
\sel{Space-Efficient Backtrack Search} we first present a generic
strategy for parallel backtrack search and then instantiate this
strategy to derive our deterministic and randomized algorithms. In
Section \sel{Space-Efficient Branch-and-Bound} we describe the
randomized parallel algorithm for branch-and-bound.  In Section
\sel{Conclusions} we give some final remarks and indicate some
interesting open problems.

\section*{Space-Efficient Backtrack Search}\label{sec:bs}
In this section we describe two parallel algorithms, a deterministic
algorithm and a Las Vegas randomized algorithm, for the backtrack
search problem.  Both algorithms implement the same strategy described
in Subsection \sel{Generic Strategy} below and require constant space per processor. The
deterministic implementation of the generic strategy (Section
\sel{Deterministic Algorithm}) requires global synchronization, while
the randomized one (Section \sel{Randomized Algorithm}) avoids
explicit global synchronization. In the rest of the paper, we let
$\Proc_0,\Proc_1, \ldots ,\Proc_{p-1}$ denote the processors in our
system.

\subsection*{Generic Strategy}\label{sec:general}
The main idea behind our generic strategy moves along the same lines
as the backtrack search algorithm of \cite{KZ93}, where at each time a
processor is either \emph{idle} or \emph{busy} exploring a certain
subtree of $\mathcal T$ in a depth-first fashion.  The computation
evolves as a sequence of \emph{epochs}, where each epoch consists of
three consecutive phases of fixed durations: (1) a \emph{traversal
phase}, where each busy processor continues the depth-first
exploration of its assigned subtree; (2) a \emph{pairing phase}, where
some busy processors are matched with distinct idle processors; and (3)
a \emph{donation phase}, where each busy processor $\Proc_i$ that was
paired with an idle processor $\Proc_j$ in the preceding phase,
attempts to entrust a portion of its assigned subtree to $\Proc_j$,
which becomes in charge of the exploration of this portion.

In \cite{KZ93} it is shown that the best progress towards completion
is achieved by letting a busy processor donate the \emph{topmost}
unexplored right subtree of the subtree which the processor 
is currently exploring. A straightforward implementation of this 
donation rule requires that a
busy processor either stores a list of up to $\BT{h}$ nodes, or, at
each donation, traverses up to $\BT{h}$ nodes in order to retrieve
the subtree to be donated, thus incurring a large time overhead.  As
anticipated in the introduction, our algorithm features a lazy
implementation of this strategy which uses constant space
per processor but incurs only a small time overhead.

We now describe in more detail how the three phases of an epoch are
performed. At any time, a busy processor $\Proc_i$ maintains the
following information, which can be stored in constant space:

\begin{itemize}
\item 
$r_i$: the root of its assigned subtree;
\item 
$v_i$: the next node to be touched by the processor in the depth-first
exploration of its assigned subtree;
\item 
$d_i \in \{\mathtt{left},\mathtt{right},\mathtt{parent}\}$: a 
\emph{direction flag} identifying the direction
where the exploration must continue after touching $v_i$.
\item 
$(t_i,q_i)$: a pair of nodes that are used to identify a portion of
the subtree to donate to an idle processor; in particular, $t_i$ is a
node on the path from $r_i$ to $v_i$, while
$q_i$ is either the right child of $r_i$ or is undefined.  We refer to
the path from $t_i$ up to $r_i$ as the \emph{tail} associated with
processor $\Proc_i$, and define the tail's length as the number of edges
it comprises.
\end{itemize}
At the beginning of the first epoch, only processor $\Proc_0$ is busy
and its variables are initialized as follows: $r_0$ is set to the root
of the tree $\mathcal T$ to be explored; $v_0=t_0=r_0$; $q_0$ is set
to the right child of $r_0$; and $d_0 = \mathtt{left}$. Consider now an
arbitrary epoch, and let $\D{t}$, $\D{p}$, and $\D{d}$ 
denote suitable values which will be fixed by the analysis.

\paragraph*{Traversal phase}
Each busy processor $\Proc_i$ advances of at most $\D{t}$ steps in the
depth-first exploration of the subtree rooted at $r_i$, starting from
$v_i$ and proceeding in the direction indicated by $d_i$. Variables
$v_i$ and $d_i$ are updated straightforwardly at each step, in accordance with the depth-first
exploration sequence. In some cases, $r_i$ and $t_i$ must also be
updated. In particular, $r_i$ is updated when $v_i=r_i$ and $d_i =
\mathtt{right}$. In this case, denoting by $w$ the right child of
$r_i$, in the next step both $r_i$ and $v_i$ are set to $w$, $d_i$ to
$\mathtt{left}$, and $q_i$ to $w$'s right child.  Instead, $t_i$ is
updated when $v_i=t_i$ and $d_i = \mathtt{parent}$. In this case, in
the next step both $t_i$ and $v_i$ are set to $v_i$'s parent. Also,
$t_i$ is updated when $t_i=r_i$ and $r_i$ is updated. In this case $t_i$ 
is set always to the new value of $r_i$. $\Proc_i$
finishes the exploration of its assigned subtree and becomes idle 
after touching $v_i$ with $v_i=r_i$ and $d_i=\mathtt{parent}$.

\paragraph*{Pairing phase}
Busy and idle processors are paired in preparation of the subsequent
donation phase. The phase runs for $\D{p}$ steps.  Different
pairing mechanisms are employed by the deterministic and the
randomized algorithm, as described in detail in the respective sections.

\paragraph*{Donation phase}

\begin{figure}
\includegraphics[width=\textwidth]{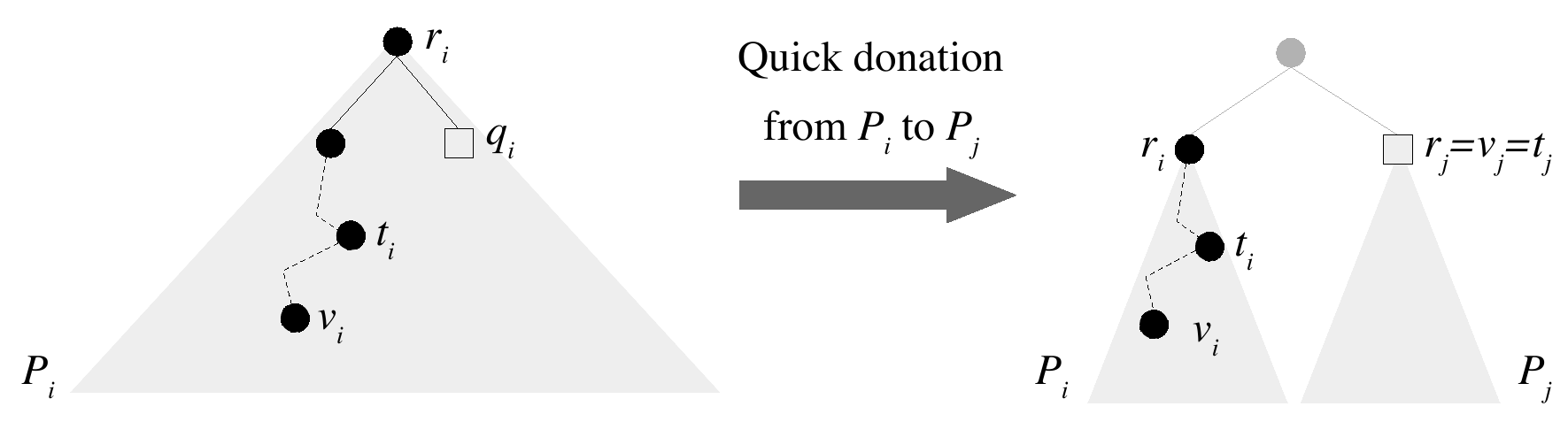}
\caption{Example of quick donation. Subtrees are denoted by shaded
area and the processors in charge of their explorations, before and
after the donation, are indicated at the bottom.  Grey circles denote
nodes that will not be touched again; black circles denote nodes have
already been touched but will be touched again; and squares denote
nodes that have not been touched yet.\label{fig:quick}}
\end{figure}

Consider a busy processor $\Proc_i$ that has been paired to an idle
processor $\Proc_j$. Two types of donations from $\Proc_i$ to
$\Proc_j$ are possible, namely a \emph{quick donation} or a \emph{slow
  donation}, depending on the status of $q_i$. As we will see, a quick
donation always starts and terminates within the same epoch, assigning
a subtree to explore to $\Proc_j$, while a slow donation may span
several epochs and may even fail to assign a subtree to $\Proc_j$.

If $q_i$ is defined, hence, it is the right child of $r_i$, a quick
donation occurs (see Figure~\ref{fig:quick}).  In this case, $\Proc_i$
donates to $\Proc_j$ the subtree rooted in $q_i$ and $\Proc_i$ keeps
the subtree rooted at the left child of $r_i$ for exploration. Thus,
$\Proc_j$ sets $r_j, v_j$ and $t_j$ all equal to $q_i$. If $q_i$ is a
leaf, then $\Proc_j$ sets $d_j$ to $\mathtt{parent}$ and $q_j$ to
undefined, otherwise it sets $q_j$ to the right child of $q_i$ and
$d_j$ to $\mathtt{left}$. Instead, $\Proc_i$ sets $r_i$ to $r_i$'s
left child and $q_i$ to undefined, while $v_i$ and $d_i$ remain
unchanged.  Also, if $t_i$ was equal to $r_i$ it is reset to the new
value of $r_i$, otherwise it remains unchanged.  (Note that quick
donation coincides with the donation strategy in~\cite{KZ93}).

If $q_i$ is undefined, a slow donation is performed where the tail
associated with $\Proc_i$ is climbed upwards to identify an unexplored
subtree which is then donated to $\Proc_j$.  To amortize the cost of
tail climbing, $\Proc_i$ attempts to donate a subtree rooted at a node
located in the middle of the tail, so to halve the length of the
residual tail that $\Proc_i$ has to climb in future slow
donations. This halving is crucial for reducing the running time.

Let us see in more detail how a slow donation is accomplished.
Initially, $\Proc_i$ verifies if a new tail must be created. This
happens if $t_i=r_i$. In this case, a \emph{tail creation} is
performed by setting $t_i = v_i$. Then, two cases are possible
depending on the tail length.

{\bf Case 1: tail length $\mathbf{\leq 1}$.}  Suppose that the tail
length is 0, that is, $t_i=v_i=r_i$. Because of the way $r_i$ is
updated in the traversal phase, it can be easily seen that in this
case $d_i = \mathtt{right}$, since otherwise $q_i$ would be defined.
Hence $\Proc_i$ must have already fully explored the left subtree of
$r_i$.  Analogously, the left subtree of $r_i$ has been fully explored
by $\Proc_i$ if $t_i$ is the right child of $r_i$ (tail length 1). In
both cases, no donation is performed and, since the current root is no
longer needed, $\Proc_i$ sets $r_i, v_i$ and $t_i$ to the right child
of $r_i$, and $d_i$ to $\mathtt{left}$. If instead, $t_i$ is the left
child of $r_i$ (tail length 1), then $\Proc_i$ donates to $\Proc_j$
the subtree rooted at the right child of $r_i$, performing the same
steps of a quick donation.  Note that in all cases, the level of the
root of the subtree assigned to $\Proc_i$ increases by 1.

\begin{figure}
\includegraphics[width=\textwidth]{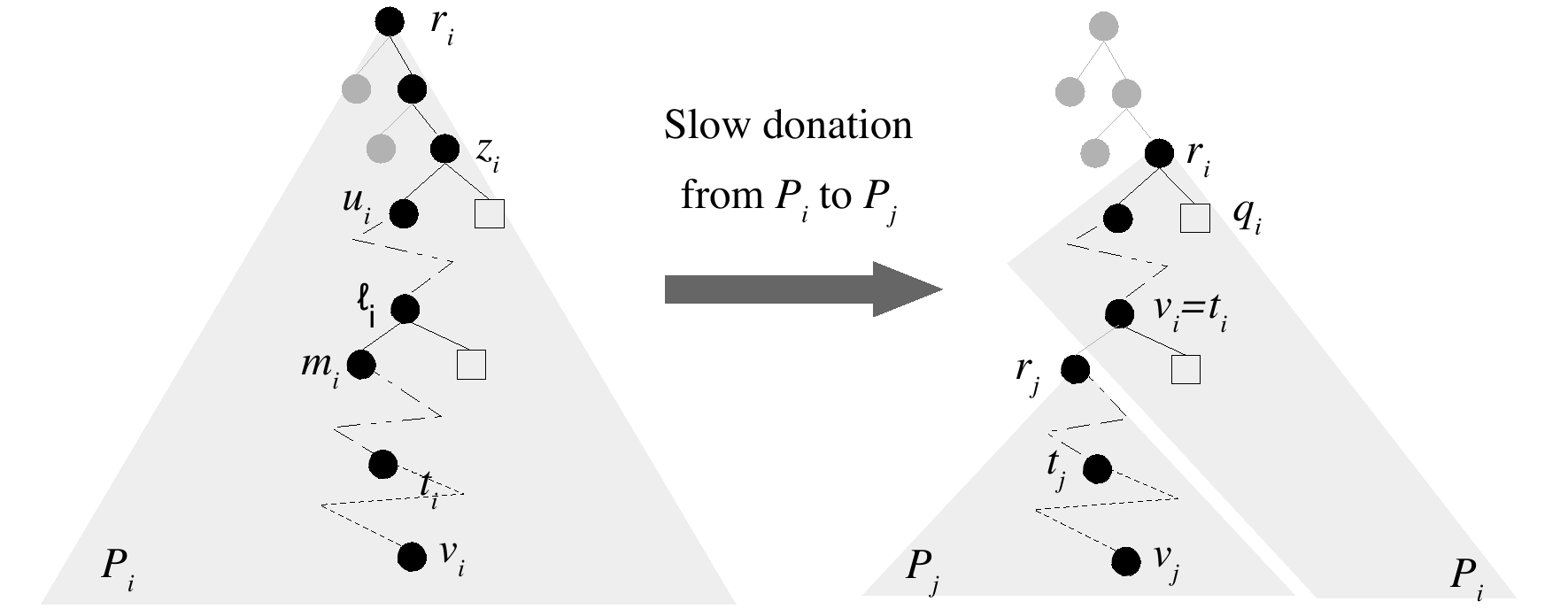}
\caption{Example of slow donation when the tail length is larger than
1. The tail is the path from $r_i$ to $t_i$. Node $u_i$ is the left
child closest to $r_i$ along the tail. Subtrees are denoted by shaded
area and the processors in charge of their explorations, before and
after the donation, are indicated at the bottom.  Grey circles denote
nodes that will not be touched again; black circles denote nodes have
already been touched but will be touched again; and squares denote
nodes that have not been touched yet.\label{fig:slow}}
\end{figure}

{\bf Case 2: tail length $\mathbf{>1}$.} 
First, processor $\Proc_i$ identifies the middle node $m_i$ of the
tail by backtracking twice from $t_i$ to $r_i$.  Then $\Proc_i$
donates to $\Proc_j$ the (partially explored) subtree rooted at $m_i$,
and $\Proc_j$ sets $r_j=m_i$, $v_j=v_i$, $d_j=d_i$, $t_j=t_i$, and
sets $q_j$ to undefined.  While backtracking to identify $m_i$,
$\Proc_i$ seeks the node $u_i$ along the path from $m_i$ to $r_i$
which is closest to $r_i$ and is the left child of its parent
$z_i$. If such a node does not exist, then $\Proc_i$ becomes idle
since only the exploration of the subtree donated to $\Proc_j$ needed
completion. If instead, $u_i$ is identified, all nodes in the path
from its parent $z_i$ (excluded) to $r_i$ (included) are unnecessary
to complete the exploration of the subtree rooted at $r_i$, since they
and their left subtrees have already been explored. Therefore
$\Proc_i$ continues the exploration by setting $r_i$to $z_i$, $q_i$ to
the right child of $z_i$, and both $v_i$ and $t_i$ to the parent
$\ell_i$ of $m_i$.  Also, $d_i$ is set to $\mathtt{right}$ if $m_i$
is the left child of $\ell_i$, or to $\mathtt{parent}$ otherwise.
This case of slow donation is depicted in Figure~\ref{fig:slow}.
Note that the level of the root of the subtree assigned to $\Proc_j$
is always greater than the level of the root of the subtree assigned
to $\Proc_i$. Moreover, the level of the root of the subtree assigned
to $\Proc_i$ either increases or remains unchanged. In this latter
case, however, $q_i$ can be set during the tail traversal so that the
next donation of $\Proc_i$ will be a quick donation.

The donation phase runs for $\D{d}$ steps, where we assume
$\D{d}$ to be greater than or equal to the maximum between the 
time for a
quick donation and the time for Case 1 of a slow donation. However,
for efficiency reasons, $\D{d}$ cannot be chosen large enough
to perform entirely Case 2 of a slow donation, since
its duration is proportional to the tail length, which may be rather large.
In this case, if $\Proc_i$ does not conclude the donation
in $\D{d}$ steps, it saves its state (requiring constant
space) at the end of the donation phase and resumes the computation
in the donation phase of the subsequent epoch, in which it 
maintains the pairing with $\Proc_j$ and refuses any further pairing. If
$t_i$ changes in the subsequent traversal phase, the state is
updated accordingly: namely, if $t_i$ is set to its father, the tail
length is updated and, if needed, $m_i$ is moved to its father to 
ensure that it remains as the middle node between $t_i$ and $r_i$. Also,
if the tail length becomes at most one, the slow donation switches
from Case 2 to Case 1.

It is easy to check that the above algorithms touches all the nodes in
the tree ${\mathcal T}$, therefore solving the backtrack search
problem.

\subsection*{Deterministic Algorithm}\label{sec:det}
In the deterministic algorithm each pairing phase is performed through
a prefix-like computation that finds a maximal matching between idle
processors and busy processors; such computation requires $\BT{\log
  p}$ parallel time. For this algorithm we set $\D{p}, \D{d} =
\BT{\log p}$, and $\D{t}=\D{d}/\kappa$, for a suitable constant
$\kappa$ defined in the proof. We call an epoch \emph{full} if in the
last step of its traversal phase at least $p/2$ processors are busy,
and we call it \emph{non-full} otherwise.
\begin{lemma}
\label{lem:full}
The total number of parallel steps in full epochs is $\BO{n/p}$.
\end{lemma}
\begin{proof}

Since each node is touched at most 3 times in a traversal phase (after
descending from the parent, after exploring the left subtree, and
after exploring the right subtree), the total number of times nodes
are touched is $\BO{n}$. The lemma follows by observing that in a full
epoch $\BT{p}$ processors touch $\BT{\log p}$ nodes each, and that the
epoch runs in $\BO{\log p}$ parallel steps.
\end{proof}

Consider an arbitrary node $q$ of ${\mathcal T}$.  Now, we bound
the number of parallel steps in non-full epochs before $q$ is touched. 
Observe that after
all leaves have been touched, the algorithm terminates in $\BO{h+\log
p}$ additional parallel steps, when all busy processors have gone back
to the roots of their assigned subtrees. 
In each epoch, we define the \emph{special processor} of $q$ as the processor exploring the
subtree containing $q$ with the deepest root; note that there is a
unique special processor in any epoch. When the special processor $S$
performs a donation to a processor $\Proc_j$, then for the subsequent
epoch either $S$ remains the special processor or $\Proc_j$ becomes
the special processor.
We denote with $\mathcal{T}_q$ the subtree of $\mathcal T$ containing nodes not larger than $q$, and with $n_q$ and $h_q$ its size and height. 

We refer to non-full epochs as \emph{donating} or \emph{preparing}
depending on the status of the special processor of $q$. Namely, a non-full
epoch is donating if the special processor $S$ completes a donation in
the epoch, while it is preparing if $S$ is involved in Case 2 of a
slow donation and, at the end of the epoch, it has not finished to
execute all operations prescribed by this type of donation.  Note
that, before $q$ is touched, any non-full epoch is always either donating or
preparing.
\begin{lemma}
\label{lem:don}
The total number of parallel steps in donating epochs 
before node $q$ is touched is $\BO{h_q \log p}$.
\end{lemma}
\begin{proof}
We claim that the level of the root of the subtree explored by special
processor $S$ increases by at least one after at most two donating
epochs. If a quick donation, or Case 1 of slow donation is performed
by $S$, then the claim is verified. Suppose $S$ is involved in Case 2
of a slow donation. Let $S=\Proc_i$ and let $\Proc_j$ be the processor
paired to $\Proc_i$. If after the donation $\Proc_i$ remains the
special processor and the root $r_i$ of its subtree is unchanged, then
during the slow donation $q_i$ has been set and hence the next donation
of the special processor is a quick donation. In all other cases,
the level root of the special processor is increased after the
donation. Thus, the claim is proved.  Since the height of $\mathcal T_q$  is $h_q$, there are $\BO{h_q}$ donating epochs and the total
number of parallel steps in donating epochs is $\BO{h_q\log p}$.
\end{proof}

We now bound the total number of parallel steps in
preparing epochs. This number is function of the number $E_q$ of parallel steps in full epochs before $q$ is touched. We observe that $E_q$ depends on $q$ but also on the the size and height of tree $\mathcal T$ (we do not represent this dependency for notational simplicity). 

\begin{lemma}
\label{lem:prep}
Let $E_q$ be the number of steps in full epochs before node $q$ is touched. Then, the total number of parallel steps in preparing epochs 
before node $q$ is touched is  $\BO{E_q+h_q \log p}$.
\end{lemma}
\begin{proof}
Consider the time interval from the beginning of the algorithm until
leaf $q$ is explored. Clearly, at any time within this interval a
special processor is defined.  We partition this interval into
\emph{eras} delimited by subsequent donation phases in which tail
creations are performed by the special processor.  (Recall that a
processor $\Proc_i$ creates a tail in the donation phase of an epoch
whenever $q_i$ is undefined, $t_i=r_i$ and $v_i \neq r_i$: then the
tail is created by setting $t_i=v_i$.) More precisely, for $i \geq 1$,
the $i$-th era begins at the donation phase of the $i$-th tail
creation, and ends right before the donation phase of the $(i+1)$-st
tail creation (or the end of the interval). Observe that the beginning
of the interval does not coincide with the beginning of the first era,
however no preparing epochs occur before the first tail creation. Note
that an era may involve more than one donation from the special
processor, and that all preparing epochs in the same era work on
segments of the tail whose creation defines the beginning of the era.
We denote with $\Phi$ the number of eras and with $\phi_i\geq 1$ the
number of slow donations in the $i$-th era, for each $1\leq i \leq
\Phi$ .

Let $T_i^j$ be the number of distinct nodes that the special processor
touches by walking up a subtree to prepare the $j$-th slow donation of
the $i$-th era, with $1\leq i \leq \Phi$ and $1 \leq j \leq \phi_i$
(nodes can be touched in both donating and preparing epochs). Since a
slow donation splits the tail in half, we have that $T_i^{j+1} \le
{T_i^{j}}/{2}$ for all $1\leq j \leq \phi_i$. Since the number of
steps in preparing epochs for one slow donation is at most proportional
to the tail length, the total time spent in preparing epochs is
$\sum_{i=1}^{\Phi} \sum_{j=1}^{\phi_i} c T_i^j \le 2c
\sum_{i=1}^{\Phi} T_i^1$, where $c\geq 1$ is a suitable
constant. We have $T_1^1\leq E_q + \BO{\log p}$ since the first era starts, in the worst case, after all full epochs and after the traversal phase of the first non-full epoch (since the special processor always receives a donation request in the first non-full epoch).

Consider an arbitrary era $i\ge2$. A node $u$ in the tail of the
era has been touched for the first time in a traversal phase of an era
$\ell<i$. Note that $\ell=i-1$ since if it was $\ell<i-1$, $u$ would
have been part of a tail created in an era before the $i$-th one and
it is easy to verify that tails of different eras are
disjoint. Therefore the number of nodes touched by the special processor
(walking upwards in the
tree) in the preparing epochs for the first donation of era $i$ is
bounded by the number of nodes touched by the special processor
in the traversal phases of era $i-1$,
which can be partitioned in three (disjoint) sets:

\begin{itemize}
\item the nodes touched for the first time in traversal phases of  full epochs in era
$i-1$; we denote the number of such nodes as $E_{i}$;
\item the nodes touched for the first time in the traversal phases of donating epochs in
era $i-1$; we denote the number of such nodes as $D_{i}$;
\item the nodes touched for the first time in the traversal phases of preparing epochs in
era $i-1$; we denote the number of such nodes as $C_{i}$.
\end{itemize}

Thus we have $\sum_{i=2}^{\Phi}  T_i^1 \le \sum_{i=2}^{\Phi} E_{i} +  \sum_{i=2}^{\Phi}
D_{i} +\sum_{i=2}^{\Phi} C_{i}$. By assumption we have $\sum_{i=2}^{\Phi} E_{i}
= E_q$, while by Lemma~\ref{lem:don} it follows that  $\sum_{i=2}^{\Phi} D_{i} =
\BO{h_q \log p}$. We now only need to bound $\sum_{i=2}^{\Phi} C_{i}$. Remember that $C_{i}$
is the number of nodes touched in the preparing epochs of the $i$-th
era that have been touched for the first time  in the traversal phases of preparing epochs of the $(i-1)$-st
era. Consider the second era: in order to bound $C_2$, we need to bound the number of
nodes that have been touched in the traversal phases of epochs in the first era. Since
$cT_{1}^j$ is an upper bound to the time required for preparing the $j$-th donation in the
first era, and since the number of nodes visited in the traversal phase of a preparing
epoch is at most a factor $1/\kappa$ the time of the respective donation phase, for a
suitable constant $\kappa$ (i.e., $\D{t}=\D{d}/\kappa$), we have $C_2 \le
\sum_{j=1}^{\phi_i} c T_1^j/\kappa \le T_1^1/2$ by setting $\kappa = 2c$. In
general, for era $i>2$ we have: $C_i \le \sum_{j=1}^{\phi_i} {cT_{i-1}^j}/{\kappa} \le
{T_{i-1}^1}/{2} \le (E_{i-1}+D_{i-1}+C_{i-1})/2$. Then, by unfolding the above 
recurrence, we get
$$
C_i \le \frac{1}{2}E_{i-1} + \frac{1}{4}E_{i-2}+\dots+\frac{1}{2^{i-2}} E_2 +
\frac{1}{2}D_{i-1} + \frac{1}{4}D_{i-2}+\dots+\frac{1}{2^{i-2}} D_2 + \frac{1}{2^{i-1}} T_1^1.
$$

Therefore, by summing up among all eras, we have
$$\sum_{i=2}^{\Phi}C_i \le  \sum_{i=1}^{\Phi} \left[ \frac{T_1^1}{2^i} +
\sum_{j=1}^{\infty} \frac{E_i}{2^j} +  
 \sum_{j=1}^{\infty} \frac{D_i}{2^j} \right]
\le T_1^1+ \sum_{i=1}^{\Phi} \left(E_i + D_i\right)
=\BO{E_q+h_q \log p}.$$
As already noticed, the number of steps in preparing epochs is proportional to the
number of nodes touched in such epochs, and this establishes the result.
\end{proof}

By combining the above three lemmas, we obtain the following theorem.
\begin{theorem}\label{th:dbt}
The deterministic algorithm for backtrack search completes in
$\BO{n/p+h\log p}$ parallel steps and constant space per processor.
\end{theorem}
\begin{proof}
By Lemma~\ref{lem:full}, there are at most $\BO{n/p}$ steps in full epochs. Then, by Lemma~\ref{lem:don} and Lemma~\ref{lem:prep} (with $E_q = \BO{n/p}$), we have that all nodes in $\mathcal{T}$ are touched after
$\BO{n/p+h \log p}$ steps. 
Let $q$ be the last touched leaf.
After $\BO{h}$ steps the processor $\Proc_i$ that have touched $q$ reaches the root $r_i$
and becomes idle since the respective subtree has been completely explored, and after
$\BO{\log p}$ steps all processors recognize the entire tree $\T$ have been explored and
the algorithm ends. Since each processor stores a constant number of words and nodes, the
theorem follows.
\end{proof}

\subsection*{Randomized Algorithm}\label{sec:rand}
In the randomized algorithm, the durations of the traversal and of the pairing phase are
set to a constant (i.e.,  $\D{d},\D{p}=\BO{1}$), and the duration of a donation
phase is set to $\D{t}=\D{d}/\kappa$, for a suitable constant $\kappa$. While the
traversal phase and the donation phase
are as described in section \sel{Generic Strategy} and are the same as in the deterministic
algorithm, the pairing phase is implemented differently as follows. In a first step, each
idle processor sends a pairing request to a random processor; in a second step, a busy
processor $\Proc_i$ that has received a pairing request from (idle) processor $\Proc_j$,
sends a message to $\Proc_j$ to establish the pairing. Note that the communication
model described in \sel{Introduction (Our Contribution)} guarantees that each busy processor receives
at most one pairing request in the first step. 
The analysis of the randomized algorithm combines elements of the analysis of the above
deterministic algorithm and the one for  the randomized backtrack search algorithm in \cite{KZ93}. 

\begin{theorem}
The randomized algorithm completes in $\BO{{n}/{p}+h}$ parallel steps with probability at
least $1 - n^{-c}$ for any constant $c > 0$.
\end{theorem}
\begin{proof}
The analysis of the randomized algorithm combines elements of the analysis of the
deterministic algorithm, presented in section \sel{Deterministic Algorithm},  and the one by Karp and
Zhang for their randomized backtrack search algorithm \cite{KZ93}. Differently from the
deterministic algorithm, we define an epoch \emph{full} if there are at least $p/4$ busy
processors at the end of the traversal phase, and \emph{non-full} otherwise. Reasoning as
in Lemma~\ref{lem:full}, we have that the number of steps in full epochs is
$\BO{{n}/{p}}$. Consider now non-full epochs and an arbitrary leaf $q$ of $\T$, and
classify such epochs as donating or preparing with respect to the special (busy) processor
$S$ for $q$. Note that due to the randomized pairing,  a non-full epoch can be both
non-donating and non-preparing with respect to $S$, since we are not guaranteed that in
non-full epochs processor $S$ is contacted by an idle processor. We call such
(non-donating and non-preparing) non-full epochs \emph{waiting} with respect to $S$.
Reasoning as in Lemma~\ref{lem:don}, we have that after $\BO{h}$ donating epochs  $q$ is
touched, hence the number of steps in donating epochs is $\BO{h}$ (recall that each epoch
comprises a constant number of steps). Also, using  the same argument of
Lemma~\ref{lem:prep}, it can be proved that the total number of steps in preparing epochs is
$\BO{n/p+ h}$. 

Finally, we upper bound the number of steps in waiting epochs by  showing that for any
leaf $q$, the number of waiting epochs before $q$ is touched is greater than $15d(n/p+h)$,
for a suitable constant $d\geq 1$, with probability at most $e^{-n/(4 p)}$. 
Consider a non-full epoch where the special processor can initiate a donation since
it does not need to complete a previously started donation. Since in a non-full epoch the number of
busy processors is $<p/4$, the number  $\hat{p}$ of idle processors that are not waiting for the
completion of donations started in earlier epochs is at least $p/2$. Therefore the probability that
the
special processor is paired to exactly one idle processor and the epoch is donating or preparing is
${\hat{p} \choose 1} (1/p)(1-1/p)^{\hat{p}-1} \ge 1/8$.
Consider now a non-full epoch where the special processor resumes a previously
interrupted donation. In this case the probability of being donating or preparing is one.
Thus, the probability that a non-full epoch is donating or preparing is at least $1/8$, while the
probability an epoch is waiting is at most $7/8$.  

Let $\B(k,N,\rho)$ denote the probability that there are less than $k$ successes in $N$
independent Bernoulli trials, where each trial has probability $\rho$ of success. As shown
above, after at most $d(n/p+h)$ donating and preparing epochs, for a suitable constant
$d\geq 1$, leaf $q$ is touched. Thus, the probability of having more than $15 d(n/p+h)$
waiting epochs is bounded above by $\B(d(n/p+h), 16d(n/p+h), 1/8)$.
By a Chernoff bound~\cite{MU05} we have that $\B(d(n/p+h),
16d(n/p+h), 1/8) \le e^{-d(n/p+h)/4} \le e^{-(d \log n)/4 }$ since $h \geq \log n$.
Then, since a waiting epoch lasts $\BO{1}$ steps, $q$ is touched after
$\BO{n/p+h}$ steps. By the union bound, the probability that each leaf is touched in
$\BO{n/p+h}$ steps is $\le n e^{-(d \log n)/4 }\leq n^{-c}$ by setting $d$ larger than $4(c+1)/\log e$. The theorem follows.

\end{proof}

\section*{Space-Efficient Branch-and-Bound}\label{sec:bb}
In this section we present a Las Vegas algorithm for the
branch-and-bound problem, which requires to explore a heap-ordered
binary tree $\T$, starting from the root, to find the minimum-cost leaf.
For simplicity, we assume that all node costs are distinct.  The
algorithm implements, in a parallel setting, an adaptation of the
sequential space-efficient strategy proposed in \cite{KSW86}, which
reduces the branch-and-bound problem to the problem of finding the
node with the $n$-th smallest cost, for exponentially increasing
values of $n$. In what follows, we first present an algorithm for a
generalized selection problem, which uses the deterministic backtrack
algorithm from the previous section as a subroutine.  The
generalization aims at controlling also the height of the explored subtrees
which, in some cases, may dominate the parallel complexity. Then, we
show how to reduce branch-and-bound to the generalized selection
problem.

\subsubsection*{Generalized Selection}\label{sec:selection}
Let $\T$ be an infinite binary tree whose nodes are associated with
distinct costs satisfying the min-heap order property, and let $c(u)$
denote the cost associated with a node $u$.  We use $\T_c$ to denote
the subtree of $\T$ containing all nodes of cost less than or equal to a
value $c$.  Given two nonnegative integers $n$ and $h$, let $c(n,h)$
be the largest cost of a node in $\T$ such that $\T_{c(n,h)}$ has at
most $n$ nodes and height at most $h$. It is important to note that
the maximality of $c(n,h)$ implies that the subtree $\T_{c(n,h)}$
must have exactly $n$ nodes or exactly height $h$ (or both).

We define the following generalized \emph{selection problem}: given
nonnegative integers $n, h$ and the root $r$ of $\T$, find the cost
$c(n,h)$.  We say that a node $u \in \T$ is \emph{good} (w.r.t. $n$
and $h$) if $c(u) \leq c(n,h)$.  Suppose we want to determine whether
a node $u$ is good. We explore $\T_{c(u)}$ using the deterministic
backtrack algorithm and keeping track, at the end of each epoch, of
the number of nodes and the height of the subtree explored until that
time.  The visit finishes as soon as the first of the following three
events occurs: (1) subtree $\T_{c(u)}$ is completely visited; or (2)
the explored subtree has more than $n$ nodes; or (3) the height of
the explored subtree is larger than $h$. Node $u$ is flagged good only when the
first event occurs. We have:
\begin{lemma}\label{lem:1}
Determining whether a node $u$ is good can be accomplished in time
$\BO{n/p+h\log p}$ using constant space per processor.
\end{lemma}
\begin{proof}
Note that keeping track of the number of nodes and the height of the
explored subtree requires minor modifications to the backtrack
algorithm and contributes an $\BO{\log p}$ additive factor to the
running time of each epoch, which is negligible. The lemma follows
immediately by applying Theorem~\ref{th:dbt} and observing that the
subtree explored to determine whether $u$ is good has at most $n$
nodes and height at most $h$.
\end{proof}

Consider a subtree $\T'$ of $\T$ with $n$ nodes and height $h$, and
suppose that some nodes of $\T'$ are marked as \emph{distinguished}.
Our selection algorithm makes use of a subroutine to efficiently pick
a node uniformly at random among the \emph{distinguished} ones of
$\T'$. To this purpose, we use reservoir sampling \cite{V85}, which
allows to sample an element uniformly at random from a data stream of
unknown size in constant space. Specifically, $\T'$ is explored using
backtrack search. During the exploration, each processor counts the
number of distinguished nodes it touches for the first time, and picks
one of them uniformly at random through reservoir sampling. The final
random node is obtained from the $p$ selected ones in $\log p$ rounds,
by discarding half of the nodes at each round, as follows. For $0\leq
k < p$, let $q^0_k$ be the number of nodes counted by processor $P_k$
in the backtrack search. In the $i$-th round, processor $P_{2^i j}$,
with $0\leq i<\log p$ and $0\leq j < p/2^{i}$, replaces its selected
node with the node selected by $P_{2^i (j+1)-1}$ with probability
$q^i_{2^i (j+1)}/(q^i_{2^i j}+q^i_{2^i(j+1)})$, and sets $q^{i+1}_{2^i
  j}$ to $q^i_{2^i j}+q^i_{2^i(j+1)}$.  After the last round, the
distinguished node held by $P_0$ is returned. We have:
\begin{lemma}\label{lem:2}
Selecting a node uniformly at random from a set of distinguished nodes
in a subtree $\T'$ of $\T$ with $n$ nodes and height $h$ can be
accomplished in time $\BO{n/p+h\log p}$, with high probability, using
constant space per processor.
\end{lemma}
\begin{proof}
We prove that at the beginning of round $i$, processor $P_{2^i j}$
contains a node selected uniformly at random from the distinguished
nodes visited by $P_{2^i j}, P_{2^i j+1}\ldots P_{2^i (j+1)-1}$ with
probability $1/q^i_{2^i j}=1/\sum_{k=2^i j}^{2^{i} (j+1)-1}q^0_k$, for
each $0\leq j < p/2^{i}$ and $0\leq i <\log p$.

The proof is by induction on $i$. At the beginning of round $0$, each processor $P_j$, with $0\leq
j<p$, contains a distinguished node sampled with probability $1/q^0_j$ by the property of
reservoir sampling. 

Suppose the claim is verified at the beginning of the round $i$, with $0\leq i<
\log p$. Then for each $0\leq j < p/2^{i}$, processor $P_{2^{i} j}$ contains a
node selected uniformly at random from the distinguished nodes touched for the first time by
$P_{2^{i} j},
P_{2^{i} j+1}, \ldots P_{2^{i} (j+1)-1}$, with
probability $1/\sum_{k=2^{i} j}^{2^{i} (j+1)-1}q^0_k$. 
The probability that $P_{2^{i} j}$ does not replace its node in round $i$ is
${q^i_{2^{i} j}}/{(q^i_{2^{i} j}+q^i_{2^{i} (j+1)})}$. Therefore, the probability a
distinguished node touched by $P_{2^{i} j}, P_{2^{i} j+1}, \ldots P_{2^{i}
(j+1)-1}$ is in $P_{2^i j}$ at end of the $i$-th round (i.e., at
the beginning of the $(i+1)$-st round) is
$$
\frac{q^i_{2^{i} j}}{q^i_{2^{i} j}+q^i_{2^{i} (j+1)}}\frac{1}{q^i_{2^{i} j}} =
\frac{1}{\sum_{k=2^i j}^{2^{i} (j+1)-1}q^0_k}=\frac{1}{q^{i+1}_{2^{i} j}}.
$$
A similar argument applies in the case  $P_{2^{i} j}$ replaces its node.
\end{proof}

We are now ready to describe the parallel algorithm for
the generalized selection problem introduced before.
The algorithm works in epochs. In the $i$-th epoch, it starts with a
lower bound $L_i$ to $c(n,h)$ (initially $L_1=-\infty$) and ends with
a new lower bound $L_{i+1}>L_i$ computed by exploring the set $F_i$
consisting of the children in $\T$ of the leaves of $\T_{L_i}$. More
in details, $L_{i+1}$ is set to the largest cost of a good node in
$F_i$ (note that if $L_i < c(n,h)$ there exists at least one good
node). The algorithm terminates as soon as $L_i = c(n,h)$.  The
largest good node in $F_i$ is computed by a binary search using random
splitters as suggested in \cite{KSW86}. The algorithm iteratively
updates two values $X^i_L$ and $X^i_U$, which represent lower and
upper bounds on the largest cost of a good node in $F_i$, until
$X^i_L=X^i_U$. Initially, we set $X^i_L=L_i$ and $X^i_U=+\infty$. The
two values are updated as follows: by using the strategy analyzed in
Lemma~\ref{lem:2}, the algorithm selects a node $u$, called
\emph{random splitter}, uniformly at random among those in $F_i$ with
cost in the range $[X^i_L, X^i_U]$ (which are the distinguished
nodes). Then, by using the strategy analyzed in Lemma~\ref{lem:1}, the
algorithm verifies if $u$ is good: if this is the case, then $X^i_L$
is set to $c(u)$, otherwise $X^i_U$ is set to $c(u)$.
\begin{theorem}\label{th:selection}
Given two nonnegative integers $n$ and $h$, the cost $c(n,h)$ in a
heap-ordered binary tree $\T$ can be determined in time $O((n /
p+h\log p)h \log n)$, with high probability, and constant space per
processor.
\end{theorem}
\begin{proof}
By Lemma~\ref{lem:1} and Lemma~\ref{lem:2}, each iteration of the
binary search algorithm requires $\BO{n/p+h\log p}$ time. Assume that
with high probability, the number of iterations of any execution of
the binary search algorithm (with random splitters) is bounded by
$K$. In this case an epoch ends in $\BO{(n/p+h\log p)K}$ time with
high probability. Consider an arbitrary leaf $q$ of $\T_{c(n,h)}$.
Clearly,  the depth of $q$ in $\T$ is at most $h$. It is easy
to see that for every $i \geq 0$, the nodes of $\T_{L_{i+1}}$ include
all those of depth $i$ or less belonging to the path from the root of
$\T$ to $q$. Therefore, after at most $h$ epochs, all leaves of
$\T_{c(n,h)}$ will be included in some $\T_{L_{i+1}}$ and the algorithm
terminates. Thus, the total time for the select algorithm is
$\BO{(n/p+h\log p)h K}$ with high probability. In what follows we
derive a bound on $K$ that holds with high probability for any
execution of the binary search algorithm.

Consider an arbitrary epoch $i$ and denote with $K$ the number of
iterations for computing $L_{i+1}$ starting from $L_i$, that is the
number of iterations to satisfy $X^i_L=X^i_U$ in the binary search
algorithm with random splitters.  For iteration $j$, with $ 0 \le j
\le K-1$, of the binary search algorithm, let $n_{i,j}$ be the number
of nodes in $F_i$ whose cost is in the range $[X^i_L,X^i_U]$. Let
$Y_{i,j}$ be a Bernoulli random variable, with $Y_{i,j}=0$ if the
random splitter $u$ in the $j$-th iteration is such that
$(1/4)n_{i,j}\leq n_{i,j+1}\leq (3/4) n_{i,j}$, and $Y_{i,j}=1$
otherwise. Note that if $Y_{i,j}=0$, $u$ partitions the $n_{i,j}$
nodes with cost in $[X^i_L, X^i_U]$ into two sets, each of cardinality
at most $3/4 n_{i,j}$. Therefore at most $\log_{4/3} n$ of the $K$
variables $Y_{i,j}, 0 \le j \le K-1$ can have value 0, that is
$\sum_{j=0}^{K-1}(1-Y_{i,j})\leq \log_{4/3} n$. Moreover, there are
$n_{i,j}/2$ nodes with cost in $[X^i_L,X^i_U]$ that can partition the
set of $n_{i,j}$ nodes with cost in $[X^i_L,X^i_U]$ into two sets of
cardinality at most $3n_{i,j}/4$, therefore $Pr[Y_{i,j}=1]=1/2$ for $0
\le j \le K-1$, and thus, by using a Chernoff bound~\cite{MU05}, we
get $ Pr[\sum_{j=0}^{K-1} Y_{i,j} \geq (1+\epsilon) K/2 ]\leq n^{-K
\epsilon^2/6}, $ for any constant $\epsilon\in(0,1]$. That is,
$\sum_{j=0}^{K-1} Y_{i,j} < (1+\epsilon) K/2$ holds with probability
$1 - n^{-K \epsilon^2/6}$, while $\sum_{i=0}^{K-1}(1-Y_{i,j})\leq
\log_{4/3} n$ always holds; combining these two events we have that
$K\leq (2/(1-\epsilon^2)) \log_{4/3}n$ with probability $1 -
n^{-(1+\BO{1})}$. Since the number of times the binary search
algorithm is executed is $\BO{n}$, by the union bound we have that
with high probability, for any execution of the binary search
algorithm $K$ is $\BO{\log n}$.
\end{proof}

We observe that by using the randomized backtrack search algorithm, the
complexity of the selection algorithm can be slightly improved.

\subsubsection*{Branch-and-Bound}\label{sec:bbalg}
Our branch-and-bound algorithm consists of a number of iterations
where we run the selection algorithm from the previous section for
exponentially increasing values of $n$ and $h$ until the first leaf is
found.  More precisely, let $n_0 = 2, h_0=1$ and $c_0 = c(n_0,h_0)$.
(Note that $c_0$ is the cost of one of the children of the root
and can be determined in constant time.)  For $i
\geq 1$, in the $i$-th iteration the algorithm determines $c_i =
c(n_i,h_i)$ where $n_i = 2n_{i-1}$, if $\T_{c_{i-1}}$ has exaclty
$n_{i-1}$ nodes, and $h_i = 2h_{i-1}$, if $\T_{c_{i-1}}$ has exaclty
height $h_{i-1}$. The loop terminates at iteration $k$, where $k$ is
the first index such that $\T_{c_k}$ includes a leaf. At that moment,
we use backtrack search to return the min-cost leaf in $\T_{c_k}$.
The following corollary is easily
established.
\begin{corollary}
\sloppy
The branch-and-bound algorithm requires $\BO{(n / p+h\log p)h \log^2
  n}$ parallel steps, with high probability, and constant space per
processor.
\end{corollary}
\begin{proof}
Consider the $i$-th iteration of the above algorithm. As observed
before, the subtree $\T_{c_{i-1}}$ must have exactly $n_{i-1}$ nodes
or exactly height $h_{i-1}$ (or both).  Hence, at least one of the two
parameters $n_i$ or $h_i$ is doubled with respect to the previous
iteration. Moreover, denoting by $n$ and $h$ the number of nodes and
height, respectively, of $\T_{c^*}$, where $c^*$ is the cost of the
minimum-cost leaf of $\T$, it is easy to show that $n_i \leq 2n$ and
$h_i \leq 2h$ for every $i$. Therefore the algorithm will execute
$\BO{\log n + \log h} =\BO{\log n}$ iterations of the selection
algorithm, and, by Theorem~\ref{th:selection}, each iteration requires
$\BO{(n / p+h\log p)h \log n}$ parallel steps. 
\end{proof}

\section*{Conclusions}\label{sec:conc}
We presented the first time-efficient combinatorial parallel search
strategies which work in constant space per processor. For backtrack
search, the time of our deterministic algorithm comes within a factor
$\BO{\log p}$ from optimal, while our randomized algorithm is
time-optimal. Building on backtrack search, we provided a randomized
algorithm for the more difficult branch-and-bound problem, which
requires constant space per processor and whose time is an $\BO{h
  \;\mbox{polylog}(n)}$ factor away from optimal.

While our results for backtrack search show that the nonconstant space
per processor required by previous algorithms is not necessary to
achieve optimal running time, our result for branch-and-bound still
leaves a gap open, and more work is needed to ascertain whether better
space-time tradeoffs can be established. However, the reduction in
space obtained by our branch-and-bound strategy could be crucial for
enabling the solution of large instances, where $n$ is  huge 
but $\BOM{n/p}$ space per processor cannot be tolerated.  
The study of space-time tradeoffs is crucial for novel
computational models such as MapReduce, suitable for cluster and cloud computing
\cite{PPRSU12}. However, algorithms for combinatorial
search strategies on such new models deserve further investigations.

As in \cite{KSW86}, our algorithms assume that the father of a tree
node can be accessed in constant time, but this feature may be hard to
implement in certain application contexts, especially for
branch-and-bound.  However, our algorithm can be adapted so to avoid
the use of this feature by increasing the space requirements of each
processor to $\BT{h}$.  We remark that even with this additional
overhead, the space required by our branch-and-bound algorithm is
still considerably smaller, for most parameter values, than that of
the state-of-the-art algorithm of \cite{KZ93}, where $\BT{n/p}$ space
per processor may be needed.

\section*{Acknowledgments}
This  work was supported, in part, by the University of Padova under Project CPDA121378, and by MIUR of Italy under project AMANDA. F. Vandin was also supported by NSF grant IIS-1247581.

\bibliographystyle{splncs} 

\bibliography{biblio}

\end{document}